\documentclass{bmcart}

\usepackage[utf8]{inputenc} 

\usepackage{booktabs}

\startlocaldefs
\usepackage{xspace}
\usepackage{amsthm}
\usepackage{amsfonts}
\usepackage{amsmath}
\usepackage{amssymb}
\usepackage{appendix}
\usepackage{color}
\newtheorem{theorem}{Theorem}
\newtheorem{lemma}[theorem]{Lemma}
\usepackage{graphicx}
\newcommand{\cX}{\mathcal{X}}
\newcommand{\cF}{\mathcal{F}}

\newcommand{\Second}{\textcolor{black}}
\newcommand{\First}{\textcolor{black}}
\endlocaldefs
\usepackage{url}

\begin{document}

\begin{frontmatter}

\begin{fmbox}
\dochead{Research}


\title{A practical approximation algorithm for solving massive instances of hybridization number for binary and nonbinary trees}


\author[
   addressref={aff1},                   
   noteref={n1},                        
   email={l.j.j.v.iersel@gmail.com}   
]{\inits{L}\fnm{Leo} \snm{van Iersel}}
\author[
   addressref={aff2},
   email={steven.kelk@maastrichtuniversity.nl}
]{\inits{S}\fnm{Steven} \snm{Kelk}}
\author[
   addressref={aff2},
         noteref={n2},                        
   email={nela.lekic@maastrichtuniversity.nl}
]{\inits{N}\fnm{Nela} \snm{Leki\'c}}
\author[
   addressref={aff3},
   email={celine.scornavacca@univ-montp2.fr}
]{\inits{C}\fnm{Celine} \snm{Scornavacca}}


\address[id=aff1]{
  \orgname{Centrum Wiskunde \& Informatica (CWI) }, 
  \street{P.O. Box 94079},                     %
  \postcode{1090 GB},                                
  \city{Amsterdam},                              
  \cny{The Netherlands}                                    
}
\address[id=aff2]{
  \orgname{Department of Knowledge Engineering (DKE), Maastricht University}, 
  \street{P.O. Box 616},                     %
  \postcode{ 6200 MD},                                
  \city{Maastricht},                              
  \cny{The Netherlands}                                    
}
\address[id=aff3]{%
  \orgname{ISEM, CNRS -- Universit\'e Montpellier II},
  \street{Place Eug\`ene Bataillon},
  \postcode{34095},
  \city{Montpellier},
  \cny{France}
}


\begin{artnotes}
\note[id=n1]{was supported by a Veni grant of The Netherlands Organisation for Scientific Research (NWO)} 
\note[id=n2]{was supported by a Vrije Competitie grant of  of The Netherlands Organisation for Scientific Research (NWO)} 
\end{artnotes}
\end{fmbox}


\begin{abstractbox}

\begin{abstract} 
\parttitle{Background}
Reticulate events play an important role in determining evolutionary relationships. The problem of computing the minimum number of such events to explain discordance between two phylogenetic trees is a hard computational problem. Even for binary trees, exact solvers struggle to solve instances with reticulation number larger than 40-50. 
\parttitle{Results}
Here we present \textsc{CycleKiller} and \textsc{NonbinaryCycleKiller}, the first methods to produce solutions verifiably close to optimality for instances with hundreds or even thousands of reticulations. 
\parttitle{Conclusions}
Using simulations, we demonstrate that these algorithms run quickly for large and difficult instances, producing solutions that are very close to optimality. As a spin-off from our simulations we also present \textsc{TerminusEst}, which is the fastest exact method currently available that can handle nonbinary trees: this is used to measure the accuracy of the \textsc{NonbinaryCycleKiller} algorithm. All three methods are based on extensions of previous theoretical work~\cite{cyclekiller,terminusest,nonbinCK} and are publicly available. \Second{We also
apply our methods to real data.}
\end{abstract}




\end{abstractbox}
%

\end{frontmatter}

\section{Background}

\Second{Phylogenetic trees are} used in biology to represent the evolutionary history of a set ${\cX}$ of species (or \emph{taxa}) \cite{MathEvPhyl,reconstructingevolution}. They are trees whose leaves are bijectively labeled by ${\cX}$ and whose internal vertices represent the ancestors of the species set; they can be rooted or unrooted. Since in a rooted tree edges have a direction, the concepts of indegree and outdegree of a vertex are well defined. \emph{Binary} rooted (phylogenetic)  trees are rooted (phylogenetic) trees whose internal vertices have outdegree~2. \emph{Nonbinary} rooted (phylogenetic) trees have no restriction on the outdegree of inner vertices.

Biological events in which a species derives its genes from different ancestors, such as hybridization, recombination and horizontal gene transfer events, cannot be modelled by a tree. To be able to represent such events, a generalization of trees is considered which allows vertices with indegree two or higher, known as \emph{reticulations}. This model, which is called a \emph{{rooted} phylogenetic network}, is of growing importance to biologists~\cite{bapteste}. For detailed background information we refer the reader to \cite{HusonRuppScornavacca10,surveycombinatorial2011,Nakhleh2009ProbSolv}.

{Although phylogenetic networks are more general than phylogenetic trees, trees are still often the basic building blocks from which phylogenetic networks are constructed. Specifically, there are many techniques available for constructing gene trees. However, when more genes are analyzed, topological conflicts between individual gene phylogenies can arise for methodological or biological reasons (e.g. aforementioned reticulate phenomena such as hybridization). This has led computational biologists to try and quantify the amount of reticulation that is needed to simultaneously explain two trees.

To state this problem more formally, \Second{we have that a  phylogenetic tree~$T$ on~$\cX$ is a refinement of a  phylogenetic tree~$T'$ on the same set~$\cX$ if $T$ can be obtained from $T'$ by deleting edges and identifying their incident vertices. Then,} 
we say that a phylogenetic network~$N$ on~$\cX$ \emph{displays} a phylogenetic tree~$T$ on~$\cX$ if~$T$ can be obtained from a subgraph of~$N$ by contracting edges. Informally, this means that (a refinement of)~$T$ can be obtained from~$N$ by, for each reticulation vertex of~$N$, ``switching off'' all but one of its incoming edges and then suppressing all indegree-1 outdegree-1 vertices (\Second{i.e. replacing paths of these vertices by one edge}). Given two rooted phylogenetic trees~$T_1$ and~$T_2$ on~$\cX$, the problem then becomes to determine the minimum number of reticulation events contained in a phylogenetic network~$N$ on~$\cX$ displaying both trees (where an indegree-$d$ reticulation counts as~$d-1$ reticulation events). The value we are minimizing is often called the \emph{hybridization number} and instead of the term phylogenetic network, the term \emph{hybridization network} is often used. It is known that the problem of computing hybridization numbers is both NP-hard and APX-hard~\cite{bordewich07a}, but it is not known whether it is in APX (i.e. whether it admits a polynomial-time approximation algorithm that achieves a constant approximation ratio).

Until recently, most research on the hybridization number of two phylogenetic trees had focused on the question of how to exactly compute this value using fixed parameter tractable (FPT) algorithms, where the parameter in question is the hybridization number~$r$ of the two trees. For an introduction to FPT we refer to~\cite{Flum2006,DowneyFellows99}.

For binary trees, algorithmic progress has been considerable in this area, with various authors reporting increasingly sophisticated FPT algorithms~\cite{bordewich2,hybridnet,quantifyingreticulation,whiddenFixed}. The fastest algorithms currently implemented are the algorithm available inside the package \textsc{Dendroscope}~\cite{Dendroscope3}, based on~\cite{fastcomputation}, and the sequence of progressively faster algorithms in the \textsc{HybridNet} family~\cite{hybridnet,chen2012algorithms,chen2013ultrafast}. The fastest theoretical FPT algorithm has running time $O( 3.18^{r}n )$ \cite{whiddenFixed}, where~$n$ is the number of taxa in the trees.

Even though in practice it rarely happens that trees are binary, the nonbinary variant of the problem has been less studied. The nonbinary version is also FPT \cite{linzsemple2009,terminusest} and a (non-FPT) algorithm has recently been implemented in \textsc{Dendroscope}~\cite{Dendroscope3}.

Such (FPT) algorithms do, however, have their limits. The running time still grows exponentially in~$r$, albeit usually at a slower rate than algorithms that have a running time of the form $n^{f(r)}$, where~$f$ is some function of~$r$. In practice this means that existing algorithms can only handle instances of binary trees when~$r$ is at most 40-50 and instances of nonbinary trees when~$r$ is at most 5-10.

These limitations are problematic. Due to ongoing advances in DNA sequencing, more and more species and strains are being sequenced. Consequently, biologists use trees with more and more taxa and software that can handle large trees is required. For such large and/or difficult trees one can try to generate heuristic or approximate solutions, but how far are such solutions from optimality? In~\cite{cyclekiller} we showed that the news is worrying. Indeed, we showed that polynomial-time constant-ratio approximation algorithms exist if and only if such algorithms exist for the problem Directed Feedback Vertex Set (DFVS). However, DFVS is a well-studied problem in combinatorial optimization and to this day it is unknown if it permits such an algorithm. Pending a major breakthrough in computer science, it therefore seems difficult to build polynomial-time algorithms which approximate hybridization number well. On the positive side, we showed that in polynomial time an algorithm with approximation ratio $O(\log r \log \log r)$ is possible. However, this algorithm is purely of theoretical interest and is not useful in practice.

\subsection{New algorithms: \textsc{CycleKiller} and \textsc{NonbinaryCycleKiller}\label{subsecAlgo}}

In this article we extend the theoretical work of \cite{cyclekiller} slightly and give it a practical twist to yield a fast approximation algorithm which we have made publicly available as the program \textsc{CycleKiller}. Furthermore, we give an implementation of the algorithm presented in~\cite{nonbinCK}, available as \textsc{NonbinaryCycleKiller}.

The worst-case running time of these approximation algorithms is exponential. However, as we demonstrate with experiments, the running time of our algorithms is in practice extremely fast. For large and/or massively discordant binary trees, \textsc{CycleKiller} is typically orders of magnitude faster than the \textsc{HybridNet} algorithms and the algorithm in \textsc{Dendroscope}. The performance gap between \textsc{NonbinaryCycleKiller} and its exact counterparts is less pronounced, but still significant, especially in its fastest mode of operation.

Of course, exact algorithms attempt to compute optimum solutions, whereas our algorithms only give approximate solutions. Nevertheless, our experiments show that when \textsc{CycleKiller} and \textsc{NonbinaryCycleKiller} are run in their most accurate mode of operation, an approximation ratio very close to~1 is not unusual, suggesting that the algorithms often produce solutions close to optimality and well within the worst-case approximation guarantee.

The idea behind the binary and nonbinary algorithm is similar. Specifically, we describe an algorithm with approximation ratio ${d(c+1)}$ for the hybridization number problem on two binary trees and an algorithm with approximation ratio ${d(c+3)}$ for the hybridization number problem on two nonbinary trees by combining a $c$-approximation for the problem MAF (Maximum Agreement Forest) with a $d$-approximation for the problem DFVS. Both these problems are NP-hard so polynomial-time algorithms attaining $c=1$ or $d=1$ are not realistic. Nevertheless, there exist extremely fast FPT algorithms for solving MAF on binary trees exactly (i.e. $c=1$), the fastest is \textsc{rSPR} by Whidden, Beiko and Zeh \cite{whiddenRSPRwebsite,whiddenRSPRexp} although the MAF algorithm inside \cite{chen2013ultrafast} is also competitive. Moreover, we observe that the type of DFVS instances that arise in practice can easily be solved using Integer Linear Programming (ILP) (and freely-available ILP solver technology such as GLPK), so $d=1$ is also often possible. 

Combining these two exact approaches gives us, in the binary case, an exponential-time approximation algorithm with worst-case approximation ratio~2 that for large instances still runs extremely quickly; this is the \texttt{2-approx} option of \textsc{CycleKiller}. In practice, we have observed that the upper bound of 2 is often pessimistic, with much better approximation ratios observed in experiments (1.003 on average for the simulations presented in this article). We find that this algorithm already allows us to cope with much bigger trees than the \textsc{HybridNet} algorithms or the algorithm in \textsc{Dendroscope}.

Nevertheless, for truly massive trees it is often not feasible to have~${c=1}$. Fortunately there exist linear-time algorithms which achieve $c=3$~\cite{whiddenFixed}. This, coupled with the fact that (even for such trees) it remains feasible to use an exact ($d=1$) solver for DFVS, means that in practice we achieve a 4-approximation for gigantic binary trees; this is the \texttt{4-approx} option of \textsc{CycleKiller}. Again, the ratio of 4  is a worst-case bound and we suspect that in practice we are doing much better than 4. However, this cannot be experimentally verified due to the lack of good lower bounds for such massive instances. In any case, the main advantage of this option is that it 
can, without too much effort, cope with trees with hundreds or thousands of taxa and hybridization number of a similar order of magnitude. An implementation of \textsc{CycleKiller} and accompanying documentation can be downloaded from \url{http://skelk.sdf-eu.org/cyclekiller}. Networks created by the algorithm can be viewed in \textsc{Dendroscope}.

For the nonbinary case, there also exist exact and approximation algorithms for MAF~\cite{nonbinCK,whiddenFixed,whiddenFixedmulti}. In case when one of the input trees is binary we can still use the exact (thus $c=1$) and approximate ($c=3$) algorithms given in \cite{whiddenFixed} (referred to as \textsc{rSPR}) to obtain respectively a 4-approximation and a 6-approximation of the hybridization number problem for nonbinary trees. When both input trees are nonbinary, then we must use the somewhat less optimized exact ($c=1$) and approximate ($c=4$) algorithms described in \cite{nonbinCK}. We then obtain 4- and 7-approximations (because in the nonbinary case $d=1$ is still easily attainable using ILP).

To measure the approximation ratios attained by \textsc{NonbinaryCycleKiller} in practice we have also implemented and made publicly available the exact nonbinary algorithm \textsc{TerminusEst}, based on the theoretical results in~\cite{terminusest}. \textsc{TerminusEst} will be of independent interest because it is currently the fastest exact nonbinary solver available.

\Second{\textsc{CycleKiller}, \textsc{NonbinaryCycleKiller} and \textsc{TerminusEst} can be downloaded respectively from \url{http://skelk.sdf-eu.org/cyclekiller} \cite{cyclekillerURL}, from \url{http://homepages.cwi.nl/~iersel/cyclekiller}  \cite{NBCKURL}, and from \url{http://skelk.sdf-eu.org/terminusest} \cite{testURL}}.

\subsection{Theoretical and practical significance}

We have described, implemented and made publicly available two algorithms with two desirable qualities: they terminate quickly even for massive instances of hybridization number and give a non-trivial guarantee of proximity to optimality. These are the first algorithms with such properties. Both algorithms are based on a non-trivial marriage of MAF and DFVS solvers (both exact and approximate), meaning that further advances in solving MAF and DFVS will directly lead to improvements in \textsc{CycleKiller} and \textsc{NonbinaryCycleKiller}.

This article also improves the theoretical work given in \cite{cyclekiller}, which also proposed using DFVS but beginning from a trivial Agreement Forest (AF) known as a \emph{chain forest}. Here we use a smarter starting point: an (approximate) MAF, and it is this insight which makes a 2-approximation (rather than the 6-approximation implied by \cite{cyclekiller}) possible when using an exact DFVS solver. Other articles have also had the idea of cycle-breaking in AFs: the advanced FPT algorithm of Whidden et al \cite{whiddenFixed} -- which has not been implemented -- and the algorithms in the aforementioned \textsc{HybridNet} family. However, both algorithms start the cycle-breaking from many starting points. In contrast, our algorithm requires only a \emph{single} starting point, i.e. a single (approximate) solution to MAF.

Here, we only present the theory behind the binary algorithm. The nonbinary case is more involved and we refer the reader to~\cite{nonbinCK} in which we introduce it. Note that our results for the binary case do not follow from the results for the nonbinary case in~\cite{nonbinCK} because here we obtain a better constant in the approximation ratio. After a presentation of the binary algorithm in Section~\ref{sec:binalg}, we will show the results of some experiments with binary trees in Section~\ref{sec:binexp} and nonbinary trees in Section~\ref{sec:nonbinexp}. \Second{Finally, in Section \ref{subsec:biodata} we
demonstrate that both \textsc{TerminusEst} and \textsc{NonbinaryCycleKiller} are easily
capable of generating optimal (respectively, nearly optimal) solutions on
a real biological dataset originally obtained from the GreenPhylDB database}.

\subsection{Technical note}

At the time the experiments on binary trees were conducted (i.e. for the preliminary version of this article  \cite{appliedcyclekiller}) \textsc{HybridNet} was the fastest algorithm available in its family. It has recently been superceded by the faster \textsc{UltraNet} \cite{chen2013ultrafast}. We believe, however, that it is neither necessary nor desirable to re-run the binary experiments, for the following reasons. In the same period the  solver
\textsc{rSPR} has also increased dramatically in speed (it is now at v1.2), leading to a corresponding speed-up in \textsc{CycleKiller}. In fact, both \textsc{rSPR} and the algorithms in the \textsc{HybridNet} family are constantly in flux and are always being improved, so any experimental setup is prone to age extremely quickly. However, the conclusions that we can derive from these experiments are unlikely to change much over time. Given that the algorithms in the \textsc{HybridNet} family (and the theoretical algorithm in
\cite{whiddenFixed}) implicitly have to explore exponentially many optimal and sub-optimal solutions to the \textsc{MAF} problem, the running time of \textsc{MAF} solvers (and thus also \textsc{CycleKiller}) is likely for the foreseeable future to remain much better than the running time of solvers for hybridization number. The central message is stable: approximating hybridization number by splitting it into \textsc{MAF} and \textsc{DFVS} instances yields extremely competitive approximation ratios for instances that exact hybridization number solvers will probably never be able to cope with.

\section{Methods}

\subsection{Preliminaries}

Let~$\cX$ be a finite set (e.g. of species). A \emph{rooted phylogenetic} $\cX$-\emph{tree} is a rooted tree with no vertices with indegree~1 and outdegree~1, a root with indegree~0 and outdegree at least~2, and leaves bijectively labelled by the elements of~$\cX$. We identify each leaf with its label and use~$L(T)$ to refer to the leaf set (or label set) of~$T$. A rooted phylogenetic $\cX$-tree is called \emph{binary} if each nonleaf vertex has outdegree two. We henceforth call a rooted, binary phylogenetic $\cX$-tree a \emph{tree} for short. For a tree~$T$ and a set $\cX'\subset \cX$, we use the notation $T(\cX')$ to denote the minimal subtree of~$T$ that contains all elements of~$\cX'$ and $T|\cX'$ denotes the result of suppressing all indegree-1 outdegree-1 vertices in $T(\cX')$.

The following definitions apply only to binary trees. Definitions for nonbinary trees are analogous but slightly more technical~\cite{nonbinCK}.

We define a \emph{forest} as a set of trees. Each element of a forest is called a \emph{component}. Let~$T$ be a tree and~$\cF$ a forest. We say that~$\cF$ is a \emph{forest for}~$T$ if, \Second{for all~$F\in\cF$}, $T|L(F)$ is isomorphic to~$F$  and the trees $\{T(L(F)),  F\in \cF\}$ are vertex-disjoint subtrees of~$T$ whose leaf-set union equals~$L(T)$. If~$T_1$ and~$T_2$ are two trees, then a forest~$\cF$ is an \emph{agreement forest} of~$T_1$ and~$T_2$ if it is a forest for~$T_1$ and~$T_2$. The number of components of~$\cF$ is denoted $|\cF|$.

We define \emph{cleaning up} a directed graph as repeatedly suppressing indegree-1 outdegree-1 vertices, removing indegree-0 outdegree-1 vertices and removing unlabelled outdegree-0 vertices until no such operation is possible. Observe that, if~$\cF$ is a forest for~$T$, $\cF$ can be obtained from~$T$ by removing~$|\cF|-1$ edges and cleaning up. From now on we consider $T_1,T_2$ as trees on the same taxon set.

\medskip
\noindent{\bf Problem:} Maximum Agreement Forest ({\sc MAF})\\
\noindent {\bf Instance:} Two rooted, binary phylogenetic trees~$T_1$ and~$T_2$.\\
\noindent {\bf Solution:} An agreement forest~$\cF$ of~$T_1$ and~$T_2$.\\
\noindent {\bf Objective:} Minimize~$|\cF|-1$.\\

The directed graph $IG(T_1,T_2,\cF)$, called the \emph{inheritance graph}, is the directed graph whose vertices are the components of~$\cF$ and which has an edge~$(F,F')$ precisely if either
\begin{itemize}
\item[$\bullet$] there is a directed path in~$T_1$ from the root of $T_1(L(F))$ to the root of $T_1(L(F'))$ or;
\item[$\bullet$] there is a directed path in~$T_2$ from the root of $T_2(L(F))$ to the root of $T_2(L(F'))$.
\end{itemize}
An agreement forest~$\cF$ of~$T_1$ and~$T_2$ is called an \emph{acyclic agreement forest} if the graph $IG(T_1,T_2,\cF)$ is acyclic. A \emph{maximum acyclic agreement forest (MAAF)} of~$T_1$ and~$T_2$ is an acyclic agreement forest of~$T_1$ and~$T_2$ with a minimum number of components.

\medskip
\noindent{\bf Problem:} Maximum Acyclic Agreement Forest ({\sc MAAF})\\
\noindent {\bf Instance:} Two rooted, binary phylogenetic trees~$T_1$ and~$T_2$. \\
\noindent {\bf Solution:} An acyclic agreement forest~$\cF$ of~$T_1$ and~$T_2$.\\
\noindent {\bf Objective:} Minimize~$|\cF|-1$.\\

We use MAF$(T_1,T_2)$ and MAAF$(T_1,T_2)$ to denote the optimal solution value of the problem {\sc MAF} and {\sc MAAF} respectively, for an instance~$T_1,T_2$.

A \emph{rooted phylogenetic network} on~$\cX$ is a directed acyclic graph with no vertices with indegree~1 and outdegree~1 and leaves bijectively labelled by the elements of~$\cX$. Rooted phylogenetic networks, which are sometimes also called hybridization networks, will henceforth be called \emph{networks} for short in this paper.  A tree~$T$ on~$\cX$ is \emph{displayed} by a network~$N$ if~$T$ can be obtained from a subtree of~$N$ by contracting edges. A \emph{reticulation} is a vertex~$v$ with~$\delta^-(v)\geq 2$ (with $\delta^-(v)$ denoting the indegree of~$v$). The \emph{reticulation number} (sometimes also called hybridization number) of a network~$N$ with root~$\rho$ is given by
$$r(N)=\sum_{v\ne\rho}(\delta^-(v)-1).$$
It was shown that the optimum to {\sc MAAF} is equal to the optimum of the following problem \cite{baroni05}.\\
\\
\noindent{\bf Problem:} {\sc MinimumHybridization}\\
\noindent {\bf Instance:} Two rooted binary phylogenetic trees $T_1$ and $T_2$. \\
\noindent {\bf Solution:} A rooted phylogenetic network~$N$ that displays $T_1$ and $T_2$.\\
\noindent {\bf Objective:} Minimize~$r(N)$.\\

Moreover, it was shown that, for two trees~$T_1,T_2$, \emph{any} acyclic agreement forest for~$T_1$ and~$T_2$ with~$k+1$ components can be turned into a phylogenetic network that displays~$T_1$ and~$T_2$ and has reticulation number~$k$, and vice versa. Thus, any approximation for {\sc MAAF} gives an approximation for {\sc MinimumHybridization}.\\

Finally, a \emph{feedback vertex set} of a directed graph is a subset of the vertices that contains at least one vertex of each directed cycle. Equivalently, a subset of the vertices of a directed graph is a \emph{feedback vertex set} if removing these vertices from the graph makes it acyclic.

\medskip
\noindent{\bf Problem:} Directed Feedback Vertex Set ({\sc DFVS})\\
\noindent {\bf Instance:} A directed graph~$D$. \\
\noindent {\bf Goal:} Find a feedback vertex set of~$D$ of minimum size.\\
\\
We note that the definition of \textsc{MinimumHybridization} easily generalises to nonbinary trees, since the definition of \emph{display} allows the image of each input tree
in the network to be more ``resolved'' than the original tree. However, the definitions of (acyclic) agreement forests are different in the nonbinary case~\cite{nonbinCK}.\\

\subsection{The algorithm for binary trees}
\label{sec:binalg}

We show how {\sc MAAF} can be approximated by combining algorithms for {\sc MAF} and {\sc DFVS}. In particular, we will prove the following theorem.

\begin{theorem}\label{thm:appr}
If there exists a $c$-approximation for {\sc MAF} and a $d$-approx-imation for {\sc DFVS}, then there exists a $d(c+1)$-approximation for {\sc MAAF} (and thus for {\sc MinimumHybridization}).
\end{theorem}

Note that this theorem does not follow from Theorem 2.1 of  \cite{nonbinCK}, since there the approximation ratio for {\sc MAAF} is a $d(c+3)$-approximation. 

To prove the theorem, suppose there exists a $c$-approximation for {\sc MAF}. Let~$T_1$ and~$T_2$ be two trees and let~$M$ be an agreement forest returned by the algorithm. Then,
\begin{equation}\label{eq:mafmaaf}
|M| -1   \leq c \cdot \text{MAF}(T_1,T_2) \leq c\cdot \text{MAAF}(T_1,T_2).
\end{equation}

An $M$-\emph{splitting} is an acyclic agreement forest that can be obtained from~$M$ by removing edges and cleaning up.

\begin{lemma}\label{lem:splitting}
Let~$T_1$ and~$T_2$ be two trees and~$M$ an agreement forest of $T_1$ and~$T_2$. Then, there exists an {$M$-splitting} of size at most $\text{MAAF}(T_1,T_2)+|M|$.
\end{lemma}
\begin{proof}
Consider a maximum acyclic agreement forest~$F$ of~$T_1$ and~$T_2$. For ${i\in\{1,2\}}$, $F$ can be obtained from~$T_i$ by removing a set of edges, say~$E_F^i$, and cleaning up. Moreover, also $M$ can be obtained from~$T_i$ by removing a set of edges, say~$E_M^i$, and cleaning up.

Now consider the forest~$S$ obtained from~$T_1$ by removing $E_M^1\cup E_F^1$ and cleaning up. Then,
\begin{itemize}
\item[$\bullet$] $S$ is an agreement forest of~$T_1$ and~$T_2$ because it can be obtained from~$T_2$ by removing edges $E_M^2\cup E_F^2$ and cleaning up;
\item[$\bullet$] $S$ is acyclic because it can be obtained by removing edges from~$F$, which is acyclic, and cleaning up;
\item[$\bullet$] $S$ can be obtained from~$M$ by removing edges and cleaning up.
\end{itemize}

Hence,~$S$ is an $M$-splitting. Furthermore, $|S| \leq |E_F^1| + |E_M^1| +1$. The lemma follows since $|E_F^1|  = \text{MAAF}(T_1,T_2)$ and $|M|= \First{|E_M^1|} +1$.
\end{proof}

Let {OptSplitting$_{T_1,T_2}(M)$} denote the size of a minimum-size $M$-splitting. Combining  Lemma~\ref{lem:splitting} and equation~\eqref{eq:mafmaaf}, we obtain
\begin{equation}\label{eq:mafmaaf2}
\text{OptSplitting}_{T_1,T_2}(M) -1 \leq (c+1) \text{MAAF}(T_1,T_2)
\end{equation}

We will now show how to find an approximation for the problem of finding an optimal $M$-splitting. We do so by reducing the problem to {\sc DFVS}. We construct an input graph~$D$ for {\sc DFVS} \First{(called the \emph{extended inheritance} graph)} as follows. For every vertex of~$M$ that has outdegree 2 (in~$M$), we create a vertex in~$D$. There is an edge in~$D$ from a vertex~$u$ to a vertex~$v$ precisely if in either~$T_1$ or~$T_2$ (or in both) there is a directed path from~$u$ to~$v$. \First{An example is in Figure~\ref{fig:example}.} We claim the following.

\begin{figure}
  \centerline{\includegraphics[scale=.6]{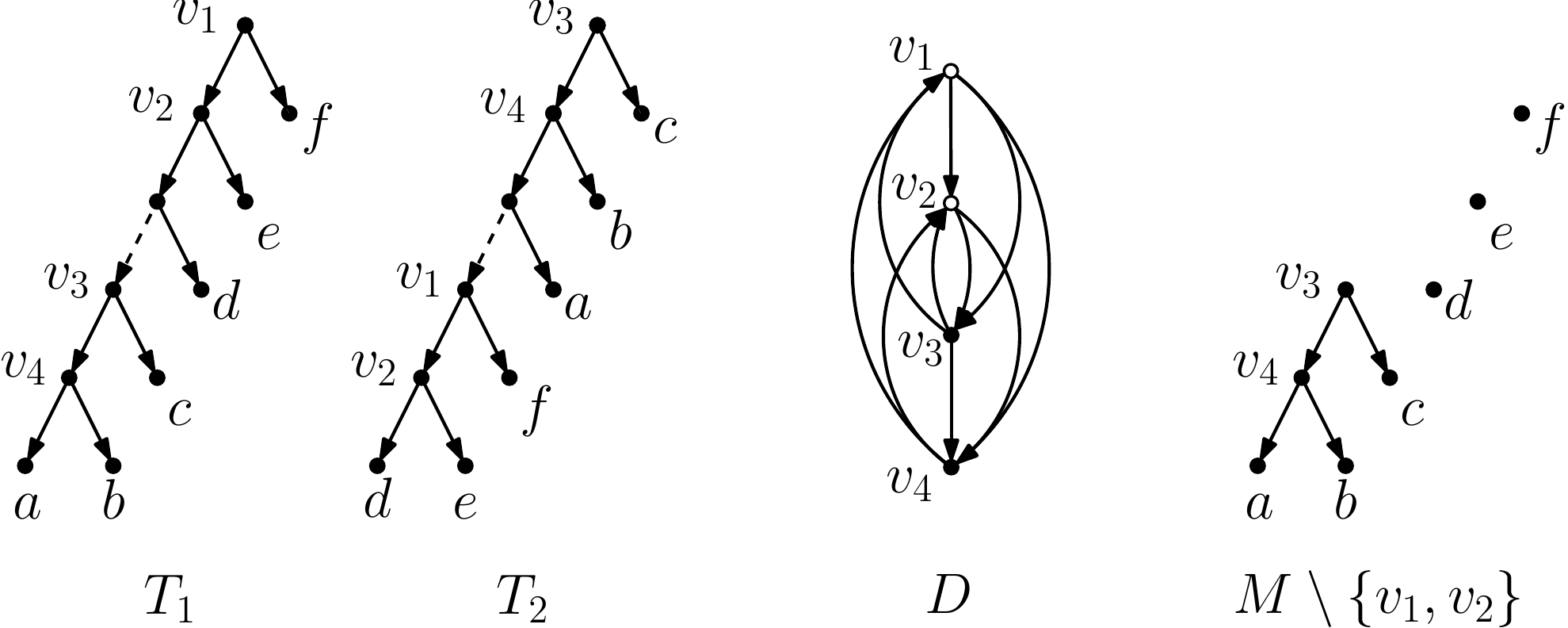}}
  \caption{\First{Two binary trees~$T_1$ and~$T_2$ and the auxiliary graph~$D$. A maximum agreement forest~$M$ of~$T_1$ and~$T_2$ is obtained by deleting the dashed edges. Graph~$D$ can be made acyclic by deleting either both filled or both unfilled vertices. Hence, removing either~$v_1$ and~$v_2$ or~$v_3$ and~$v_4$ from~$M$ makes it an acyclic agreement forest for~$T_1$ and~$T_2$, see Lemma~\ref{lem:dfvs}. The acyclic agreement forest~$M\setminus\{v_1,v_2\}$ obtained by removing~$v_1$ and~$v_2$ from~$M$ is depicted on the right.}}
  \label{fig:example}
\end{figure}

\begin{lemma}\label{lem:dfvs}
A subset~$V'$ of the vertices of~$D$ is a feedback vertex set of~$D$ if and only if removing~$V'$ from~$M$ makes it an acyclic agreement forest.
\end{lemma}
\begin{proof}
We show that~$D\setminus V'$ has a directed cycle if and only if the inheritance graph of $M\setminus V'$ has a directed cycle.

To prove this, first suppose that there is a cycle $v_1,v_2,\ldots,v_k =
v_1$ in the inheritance graph of $M\setminus V'$. The vertices in the inheritance
graph of $M\setminus V'$ correspond to the roots of the components of $M\setminus V'$. Since these roots have outdegree 2 in $M\setminus V'$, they had outdegree 2 in $M$, and
are thus vertices of $D$. So the vertices $v_1,v_2,\ldots,v_k$ that form the
cycle are vertices of $D$. Since these vertices are in the inheritance
graph of $M\setminus V'$, they can not be in $V'$ and so they are vertices of $D\setminus V'$. The reachability relation between these vertices in $D\setminus V'$ is the same
as in the inheritance graph of $M\setminus V'$. So, the vertices $v_1,v_2,\ldots,v_k$
form a cycle in $D\setminus V'$.

Now suppose that there is a cycle $w_1,w_2,\ldots,w_k = w_1$ in $D\setminus V'$. Each
of the vertices $w_1, w_2,\ldots, w_k$ is a vertex with outdegree-2 in $M$.
Some of them might be roots of components, while others are not.
However, observe that if there is a directed path from a vertex $u$ to a
vertex $v$ in $T_1$ (or in $T_2$) then there is also a directed path from the root of
the component of $M\setminus V'$ that contains $u$ to the root of the component of
$M\setminus V'$ that contains $v$. Hence, there is a directed cycle in the
inheritance graph of $M\setminus V'$, formed by the roots of the components of
$M\setminus V'$ that contain $w_1,w_2,\ldots,w_k$.
\end{proof}

\begin{proof}[\First{Proof of Theorem \ref{thm:appr}}]
Suppose that there exists a $d$-approximation for {\sc DFVS}. Let~FVS be a feedback vertex set returned by this algorithm and let~$\text{MFVS}$ be a minimum feedback vertex set. Then, removing the vertices of~$\text{MFVS}$ from~$M$ gives an optimal $M$-splitting. Furthermore, {OptSplitting$_{T_1,T_2}(M) = |M|+|\text{MFVS}|$}. This is because for every vertex in a cycle $C$, its parent in~$M$ must participate in some cycle that contains elements of $C$. So if we start by removing the root of the component we are splitting and subsequently remove only those vertices whose parents have already been removed we see that we add at most one component per vertex. In fact, because vertices of~$D$ all have out-degree~2 in~$M$, we add exactly one component per vertex.

By removing the vertices of~FVS from~$M$, we obtain an acyclic agreement forest~$\cF$ such that
\begin{eqnarray*}
|\cF| - 1 & = & |M| + |\text{FVS}| - 1\\
& \leq & |M| + d\cdot |\text{MFVS}|  - 1\\
& \leq & d (|M| +|\text{MFVS}| - 1)\\
& = & d (\text{OptSplitting}_{T_1,T_2}(M) - 1)\\
& \leq & d (c+1) \text{MAAF}(T_1,T_2),
\end{eqnarray*}
where the last inequality follows from equation~\eqref{eq:mafmaaf2}. Thus,~$\cF$ is a $d(c+1)$-approximation to {\sc MAAF}, which concludes the proof of Theorem~\ref{thm:appr}.
\end{proof}

\First{Theorem \ref{thm:appr} implies that a solution to the MAAF problem for a given instance can be constructed by (i)   finding a solution $\mathcal{F}$  to the MAF problem for the same instance  (ii) constructing the  extended inheritance graph $D$ for $\mathcal{F}$ (iii) finding a solution $V$ for the DFVS problem on the graph $D$ and (iv) modifying $\mathcal{F}$ accordantly to $V$.}
\section{\Second{Experiments} and discussion}

\subsection{Practical experiments with binary trees}
\label{sec:binexp}
To assess the performance of {\sc CycleKiller}, a simulation study was undertaken. We generated 3 synthetic datasets, an \emph{easy}, a \emph{medium} and a \emph{hard} one, containing respectively 800, 640 and 640 pairs of rooted binary phylogenetic trees. 

The easy data set was created by varying  two parameters, namely the number of taxa $n$ and the number of rSPR-moves $k$ used to obtain the
second tree from the first (note that this number is an upper bound on the actual rSPR distance). The 800 pairs of rooted binary phylogenetic trees were created by varying  
$n$ in $\{20,50,100,200\}$ and $k$ in $\{5,10,...,25\}$, and then creating 40 different instances per each combination of parameters.
Each pair $(T_1,T_2)$ of  rooted binary phylogenetic trees for a given set of parameters  $n$ and $k$ is created as follows: 
The first tree $T_1$ on $\mathcal{X}=\{x_1,\dots,x_n\}$ is generated  by first creating a set of  $n$ leaf vertices bijectively labeled by the set $\mathcal{X}$. Then, 
two vertices $u$ and $v$, both with indegree 0, are randomly
picked  and  a new vertex $w$, along with two new edges $(w,u)$ and $(w,v)$, is created. This is done until only one vertex with no ancestor, the root, is present. The  second tree $T_2$ is  obtained from $T_1$ by applying $k$ rSPR-moves.
The medium and the hard data sets were generated in the same way as the easy one, but for different choices of the parameters:   $n$ in $\{50,100,200,300\}$ and $k$ in $\{15,25,40,55\}$ for the medium one and $n$ in $\{100,200,400,500\}$  and $k$ in $\{40,60,80,100\}$ for the hard one.

The exact hybridization number has been computed by \textsc{HybridNet} \cite{hybridnet}, available from \url{http://www.cs.cityu.edu.hk/~lwang/software/Hn/treeComp.html} or with  \textsc{Dendroscope} \cite{Dendroscope3}, available from \url{http://www.dendroscope.org}. We will refer to these algorithms as the \emph{exact algorithms}. Each instance has been run on a single core of an Intel Xeon E5506 processor.

Each  run that took more than one hour was aborted. For each instance, we ran our program with the option \texttt{2-approx}, and, in case the latter did not finish within one hour, we ran it again, this time using the option \texttt{4-approx}, always with a one-hour limit (\Second{see Section \ref{subsecAlgo}}). We used
the program \textsc{rSPR} v1.03 \cite{whiddenRSPRwebsite,whiddenRSPRexp} to solve or approximate MAF
and GLPK v4.47 (\url{http://www.gnu.org/software/glpk/}) to solve the \textcolor{black}{following simple polynomial-size ILP formulation of \textsc{DFVS}:}
\textcolor{black}{
\begin{align*}
\min & \text{  }\sum_{v\in V} x_v &\\
\text{s.t.}\\
& 0 \leq \ell_{v} \leq |V| - 1     				&&\text{for all } v \in V \\
& \ell_{v} \geq \ell_{u} + 1 - |V|x_u - |V|x_v			&&\text{for all } e=(u,v)\in E \\ 		
& \ell_{v} \in \mathbb{Z} 					&&\text{for all } v \in V \\         
& x_{v}\in\{0,1\}     						&&\text{for all } v \in V 
\end{align*}
%
}
\textcolor{black}{Given a directed graph $D=(V,E)$, the binary variables $x_v$ model whether
a vertex is in the feedback vertex set, and the integer variables $\ell_v$ model the positions
of the surviving vertices in the induced topological order. The edge constraints enforce the
topological order. Note that an edge constraint is essentially eliminated if one or both endpoints of the edge are in the feedback vertex set}.

For all instances of the easy data set, {\sc CycleKiller} finished with the \texttt{2-approx} option within the one hour limit, while for 33 instances the exact algorithms were unable to compute the hybridization number. Note that, even for ``easy'' instances, computing the exact hybridization number can take a very long time. To give the reader an idea, for 9 runs of the easy data, \textsc{Dendroscope} and \textsc{HybridNet} did not complete within 10 days. Table~1 shows a summary of the results. It can be seen that {\sc CycleKiller} was much faster than the exact algorithms. Moreover, for 96.6\% of the instances for which an exact algorithm could find a solution, {\sc CycleKiller} also found an optimal solution. While the theoretical worst-case approximation ratio of the \texttt{2-approx} option of {\sc CycleKiller} is~2, in our experiments it performed very close to a 1-approximation.

For the medium data set, {\sc CycleKiller} finished with the \texttt{2-approx} option for 613 instances, and for the remaining ones with the \texttt{4-approx} option. The exact algorithms could compute the hybridization number for only 199 instances (out of 640). For 97.5\% of these instances, {\sc CycleKiller} also found an optimal solution, but with a much better running time.
Regarding the hard data set, 444 runs were completed with the \texttt{2-\hspace{0pt}approx} option and for the remaining ones we were able to use the \texttt{4-approx} option within the given time constraint. Unfortunately, the exact algorithms were unable to compute  the hybridization number for any tree-pair of this data set and hence we could not compute the average approximation ratios.
Over all our experiments, the maximum hybridization number that the exact algorithms could handle was 25.\footnote{In \cite{fastcomputation}, it has been shown that this number can go up to 40 when running Dendroscope on a similar processor but allocating all cores for one instance, i.e. exploiting the possibilities of parallel computation of  this implementation.} In contrast, the \texttt{2-approx} option of {\sc CycleKiller} could be used for instances for which the size of a MAF was up to 97, and thus for instances for which the hybridization number was at least 97.
 
To find the limits of the \texttt{4-approx} option of \textsc{CycleKiller}, we also tested it on randomly generated trees. On a normal laptop, it could construct networks with up to 10,000 leaves and up to 10,000 reticulations within 10 minutes. Since the number of reticulations found is at most four times the optimal hybridization number, this implies that the \texttt{4-approx} option of \textsc{CycleKiller} can handle hybridization numbers up to at least 2,500. These randomly generated trees are, however, biologically meaningless and, therefore, we conducted the extensive experiment described above on trees generated by rSPR moves. Finally we note that over all experiments the worst approximation ratio we encountered was 1.2.


\subsection{Practical experiments with nonbinary trees}
\label{sec:nonbinexp}

To run the simulations with \textsc{NonbinaryCycleKiller}, we used a subset of the trees from the easy set of binary experiments. We then applied random edge contractions in order to obtain nonbinary trees. Hence, we have the same two parameters as before, namely the number of taxa $n \in\{20, 50, 100\}$ and the number of rSPR-moves $k \in \{5, 10, 15, 20\}$, and an additional parameter $\rho \in \{25, 50, 75\}$ which measures the percentage of the edges of an original binary tree that were contracted in order to obtain a nonbinary tree. We could only use smaller values of~$n$ and~$k$ from the easy set of experiments because exact solvers for nonbinary MAF (upon which \textsc{NonbinaryCycleKiller} is built) and exact solvers for nonbinary \textsc{MinimumHybridization} (which is important to measure the accuracy of \textsc{NonbinaryCycleKiller} in practice) are slower than their binary counterparts.

We performed two runs of experiments.\footnote{\textcolor{black}{We note that \textsc{NonbinaryCycleKiller} uses a row-generation ILP formulation - based on \cite{dfvsApprox} - to solve DFVS, rather
than the polynomial-size formulation used by \textsc{CycleKiller}. ILP is in neither case a bottleneck for the running time.}} One run with instances consisting of one binary and one nonbinary tree, and one run with instances consisting of two nonbinary trees.  

For the experiments with one binary and one nonbinary tree, we were still able to use the \textsc{rSPR} algorithm \cite{whiddenRSPRwebsite,whiddenFixed}, which has a better running time and approximation ratio compared to the available algorithm for two nonbinary trees. When \textsc{rSPR} is used in exact mode, \textsc{NonbinaryCycleKiller} yields a theoretical worst-case approximation ratio of 4. When \textsc{rSPR} is used in its 3-approximation mode, \textsc{NonbinaryCycleKiller} yields a theoretical worst-case approximation ratio of~6 (\Second{see Section \ref{subsecAlgo}}). The results of this run are summarized in Table \ref{table:rspr_for_maf}.

For the experiments with two nonbinary trees, the \textsc{rSPR} software can no longer be used, and instead we used the exact and 4-approximate MAF algorithm described in~\cite{nonbinCK}. This makes \textsc{NonbinaryCycleKiller} behave as a 4-approximation and 7-approximation respectively (\Second{see Section \ref{subsecAlgo}}). Note that the exact algorithm~\cite{nonbinCK} is considerably slower than \textsc{rSPR}, meaning that in practice \textsc{NonbinaryCycleKiller} struggles with two nonbinary trees more than
when at most one of the trees is nonbinary. The results for this run are summarized in Table~\ref{table:maf_for_maf}.

The exact hybridization number in both runs was computed by \textsc{TerminusEst}~\cite{testURL}.

Each instance that took longer than~10 minutes to compute was aborted and the running time was set to~600 seconds. The averages of the running-times are taken over all instances, with running-time taken to be~600 if the program timed out for that instance. (We used a shorter time-out than in the binary experiments because of the observation that, in the nonbinary case, exact algorithms running longer than~10 minutes almost always took longer than~60 minutes too.)

Note that we did not compare the performance of \textsc{NonbinaryCycleKiller} to \textsc{Dendroscope} because \textsc{TerminusEst} has better running times than the exact nonbinary \textsc{MinimumHybridization} solver inside \textsc{Dendroscope} (data not shown).

To enable a clearer analysis we divided the trees into representative ``simple'' and ``tricky'' ones based on two parameters,~$n$ and~$k$. Parameter values for the simple set were $n \in \{20, 50\}$, $k \in \{5, 10, 15\}$ and for the tricky set $n \in \{50, 100\}$, $k=20$. In addition we varied the percentage of contracted edges (in a single tree in the first run and in both trees in the second run).

In Table~\ref{table:rspr_for_maf} we show running times and solution quality of our algorithm when one of the input trees is binary. For the simple set of instances (regardless of the percentage of edge-contractions) we see that the more accurate version of our algorithm, the 4-approximation, had a better running time than the exact algorithm, and at the same time had an average approximation ratio very close to~1. Far more interesting is to see what happens with tricky instances. As predicted, the running time of the exact algorithm is much higher for tricky instances due to the higher hybridization numbers. On the other hand, the running time of the 4-approximation does not rise significantly at all, whilst still attaining an approximation ratio again very close to 1. Another thing to note is that the percentage of contraction only seems to affect the running time of the exact algorithm. The practical worst-case approximation ratio observed in these experiments was~1.75 for the 4-approximation and~3 for the 6-approximation.

Table \ref{table:maf_for_maf} shows our results on instances with two nonbinary trees. The exact algorithm for MAF is in this case much slower and this affects the running times even for the simple set. While the 4-approximation version has an average approximation ratio very close to 1 again, the running time is in this case worse than that of \textsc{TerminusEst}. For the tricky set the situation is even more significant; the exact MAF algorithm cannot deal with reticulation numbers above~15, while \textsc{TerminusEst} can get slightly further. On the other hand, the 7-approximation still runs much faster than \textsc{TerminusEst}, both for simple and tricky instances, while having an average approximation ratio of less than 2.6. The practical worst-case approximation ratio observed in these experiments was~1.5 for the 4-approximation and~4 for the 7-approximation. 


It is worth noting that, for the 4-approximation, the running time for the 75\%-contraction trees is considerably lower than the one for the 50\%-contraction trees. This is due to the fact that a high contraction in both trees causes the hybridization number of the instance to drop, and a lower hybridization number leads to a better running time. 
Also note that the exact solver \textsc{TerminusEst} seems more able to cope with the
tricky 25\%-contraction instances than the tricky 50\%-contraction instances. This is probably because, although low contraction rates yield a higher hybridization number, the trees remain ``relatively binary'' and this can induce more efficient branching in the underlying FPT algorithm~\cite{terminusest}. It is plausible that with 50\%-contraction the instances suffer from the disadvantage of relatively high hybridization number without the branching advantages associated with (relatively) binary trees.

To find the limits of the \texttt{7-approx} option of \textsc{NBCK}, we also tested it on huge, biologically meaningless, randomly generated trees. Below some results:
\begin{itemize}
\item [$\bullet$]1000 leaves, 25\%-contraction, on average 995 reticulations in 63 sec.
\item [$\bullet$] 1000 leaves, 50\%-contraction, on average 989 reticulations in 82 sec.
\item [$\bullet$] 1000 leaves, 75\%-contraction, on average 840 reticulations, in 656 sec. 
\end{itemize}
Computation times of this last run of experiments do not include the network construction.

\Second{\subsection{Practical experiments on biologically relevant trees}
\label{subsec:biodata}
Finally, we tested our methods on phylogenetic trees obtained from GreenPhylDB  \cite{Rouard22092010} -- version 3, a database containing  twenty-two full genomes of members of the plantae kingdom, ranging from algae to angiosperms. We were able  to retrieve from the database the
9903 rooted phylogenetic trees associated to the gene families contained in the database (the gene trees), along with the rooted phylogenetic tree describing 
the history of the  twenty-two species contained in GreenPhylDB (the species tree). Note that the species tree for  these species is not completely resolved, i.e. it is nonbinary. Among the gene trees, 2769 contain less than 3 species and they were discarded. Of the remaining 7134 trees, only 204 were directly usable for testing our methods.
Indeed, because of gene duplication events arising in genomes, some species host several copies of the same gene, hence individual gene trees usually have several leaves labeled with identical species names.
Unfortunately, our methods do not handle such multi-labeled gene trees (MUL trees). We thus transformed  the MUL trees into trees containing single copies of labels, applying the tools described in \cite{Scornavacca2011,SSIMULURL} to the forest $F$ of 7134 trees. As in Section 4.1 of \cite{Scornavacca2011}, we obtained four data sets: 
$F_1$, $F_2$, $F^p_3$ and $F^s_3$, respectively containing 204, 1003, 5924 and 5789 trees.
Note that only $F^s_3$ contains nonbinary trees. Finally, for each single labeled tree $G \in (F_1 \cup F_2 \cup F^p_3 \cup F^s_3)$, we restricted the species tree $S$ (containing 22 taxa) to the leaves of $G$ and we applied our methods to all so obtained pairs (restricted $S$,$G$). The results are presented in Tables \ref{lab:f1} - \ref{lab:f3s}. For all four datasets both \textsc{TerminusEst} and
\textsc{NonbinaryCycleKiller} ran extremely quickly, rarely taking more than a couple of seconds for each species-gene tree pair. Moreover, the clear conclusion with this dataset is that, although
the species-gene pairs are often incompatible, there are rarely many cycles to kill and
optimum solutions to the hybridization number problem are generally extremely close to
optimal solutions to MAF.}

\section{Conclusions}

Our experiments with binary trees show that \textsc{CycleKiller} is much faster than available exact methods once the input trees become sufficiently large and/or discordant. In over $96\%$ of the cases \textsc{CycleKiller} finds the optimal solution and in the remaining cases it finds a solution very close to the optimum. We have shown that the most accurate mode of the program produces solutions that are at most a factor~2 from the optimum. In practice, the average-case approximation ratio that we observed was~1.003. The fastest mode of the algorithm can be used on trees with thousands of leaves and provably constructs networks that are at most a factor of~4 from the optimum.

Our experiments with nonbinary trees highlight once again that the cycle-breaking technique described in this article is intrinsically linked to the current state-of-the-art in \textsc{MAF} algorithms. \textsc{TerminusEst} is faster than the most accurate mode of \textsc{NonbinaryCycleKiller} when both trees are nonbinary due to the fact that MAF solvers for two nonbinary trees have not yet been optimized to the same extent as their binary counterparts. In fact, \textsc{TerminusEst} is the best avaible exact method for nonbinary trees and can handle instances for which the optimum is up to~15-20. For other instances, \textsc{NonbinaryCycleKiller} in its fastest mode is much faster than \textsc{TerminusEst} and produces solutions that are at most a factor~4 from the optimum (less than 2.6 on average).

Finally, for instances with one binary and one nonbinary tree, the most accurate mode of \textsc{NonbinaryCycleKiller} is again much faster than \textsc{TerminusEst} and produces solutions that are at most a factor~1.75 from the optimum (less than 1.011 on average). 


\begin{backmatter}

\section*{Competing interests}
  The authors declare that they have no competing interests.

\section*{Author's contributions}
All the authors conceived the ideas, designed and conducted the experiments, and wrote and approved the paper. 

\section*{Acknowledgements}
A preliminary version of this paper (restricted to the binary case) appeared in the proceedings of the 12th Workshop on Algorithms in Bioinformatics (WABI 2012)~\cite{appliedcyclekiller}.
We thank Simone Linz and Leen Stougie for fruitful discussions. 
{This publication is the contribution no. 2014-040 of the Institut des Sciences de l'Evolution de Montpellier (ISE-M, UMR 5554). This work has been partially funded by the French {\it Agence Nationale de la Recherche}, {\it Investissements d'avenir/Bioinformatique} (ANR-10-BINF-01-02, {\it Ancestrome}), and it has benefited from the ISE-M computing facilities.

}

\bibliographystyle{bmc-mathphys}  
\bibliography{BMC_revision_v4}      




%



\newpage

\section*{Tables}

\begin{table}[h!]
\begin{center}
\begin{tabular}{lr|r|r|r|r|r|r|r|r|}
\cline{3-10}
& & \multicolumn{2}{|c|}{Exact algorithms} & \multicolumn{6}{|c|}{\sc CycleKiller}\\
\hline

 \multicolumn{1}{|l|}{Dataset} & \multicolumn{1}{|l|}{Total} & Com- & \multicolumn{1}{|l|}{Time} &  \multicolumn{2}{|c|}{\texttt{2-approx}}  &  \multicolumn{2}{|c|}{\texttt{4-approx}}&  \multicolumn{1}{|l|}{Ratio} &  \multicolumn{1}{|l|}{Opt.}\\
 \cline{5-6}
   \cline{7-8}
 \multicolumn{1}{|l|}{} & \multicolumn{1}{|l|}{runs} & \multicolumn{1}{|l|}{pleted} &  \multicolumn{1}{|l|}{} & \multicolumn{1}{|l|}{Compl.}  &  \multicolumn{1}{|c|}{Time}  & \multicolumn{1}{|l|}{Compl.}  &  \multicolumn{1}{|c|}{Time} & \multicolumn{1}{|c|}{} &  \multicolumn{1}{|l|}{found}\\
\hline
\multicolumn{1}{|l|}{Easy} & 800 & 767 & 798 & 800 & 3 & -  & -  & 1.003 & 96.6\% \\
\multicolumn{1}{|l|}{Medium} & 640 & 199 & 2572 & 613  & 212 & 27 & $<$1  & 1.002 & 97.5\%\\
\multicolumn{1}{|l|}{Hard} & 640 & 0 & 3600 & 440  &  1271& 200 &1.5 & - & -\\
\hline
\end{tabular}
\vspace{0.25cm}
\end{center}
\caption{Experimental results for instances with two binary trees. The third column indicates for how many instances at least one exact algorithm finished within one hour. The fifth column indicates for how many instances the \texttt{2-approx} option of {\sc CycleKiller} finished within one hour. For the remaining instances, the \texttt{4-approx} option finished within one hour, as can be seen from the seventh column. The average running time for the \texttt{2-approx} and the \texttt{4-approx} in seconds are reported respectively in the sixth and eighth column. The average approximation ratio (ninth column) is taken over all instances for which at least one exact method finished.  The last column indicates the percentage of those instances for which {\sc CycleKiller} found an optimal solution.\label{table:experiments}}
\end{table}

\begin{table}[h!]
 \begin{center}
  \begin{tabular}{l l|r|r|r|r|r|r|r|r|}
   \cline{3-10}
     & & \multicolumn{2}{|c}{\textsc{TerminusEst}} & \multicolumn{6}{|c|}{\textsc{NonbinaryCycleKiller}}\\
     \cline{5-10}
     & & \multicolumn{2}{|c}{} &\multicolumn{3}{|c|}{\texttt{4-approx}} & \multicolumn{3}{|c|}{\texttt{6-approx}}\\
   \hline
     \multicolumn{1}{|l|}{Contr.} & \multicolumn{1}{|l|}{Dataset} & $opt$ & Time & $r(N)$ & Time & Ratio & $r(N)$ & Time & Ratio \\
   \hline
     \multicolumn{1}{|l|}{25\%} & Simple &  7.504	&8.004	&7.567	&0.967	&1.007	&11.421	&0.996	&1.532 \\
     \multicolumn{1}{|l|}{}     & Tricky & 17.000	&203.650	&17.288	&3.675	&1.003	&27.238	&3.638	&1.600 \\
   \hline
     \multicolumn{1}{|l|}{50\%} & Simple & 6.736	&9.896	&6.829	&0.942	&1.008	&10.900	&0.925	&1.639 \\
     \multicolumn{1}{|l|}{}     & Tricky & 14.976	&374.263	&16.288	&3.388	&1.006	&26.413	&3.438	&1.640 \\ 
   \hline
     \multicolumn{1}{|l|}{75\%} & Simple & 5.139	&12.304	&5.263	&0.867	&1.011	&8.692	&0.963	&1.659 \\
     \multicolumn{1}{|l|}{}     & Tricky & 10.500	&391.575	&13.475	&3.263	&1.006	&23.200	&3.275	&1.633 \\
   \hline
	 \hline
	  \multicolumn{2}{|c|}{Worst case}  &20	&600	&22 &15 &1.75	&37	&13	&3 \\
	 \hline
  \end{tabular}
  \vspace{.1cm}
 \end{center}
 \caption{Summary of results for instances with one binary and one nonbinary tree. We list the average hybridization number found ($opt$ and $r(N)$), the average running time in seconds (Time) and where applicable the average approximation ratio (Ratio) for the three algorithms.\label{table:rspr_for_maf}}
\end{table}

\begin{table}[h!]
 \begin{center}
   \begin{tabular}{l l|r|r|r|r|r|r|r|r|}
   \cline{3-10}
     & & \multicolumn{2}{|c}{\textsc{TerminusEst}} & \multicolumn{6}{|c|}{\textsc{NonbinaryCycleKiller}}\\
     \cline{5-10}
     & & \multicolumn{2}{|c}{} &\multicolumn{3}{|c|}{\texttt{4-approx}} & \multicolumn{3}{|c|}{\texttt{7-approx}}\\
   \hline
     \multicolumn{1}{|l|}{Contr.} & \multicolumn{1}{|l|}{Dataset} & $opt$ & Time & $r(N)$ & Time & Ratio & $r(N)$ & Time & Ratio \\
   \hline
     \multicolumn{1}{|l|}{25\%} & Simple & 7.168	&12.971	&7.240	&43.967	&1.032	&16.338	&2.463	&2.343  \\
     \multicolumn{1}{|l|}{}     & Tricky & 16.148	&279.100	&-	&-	&-	&35.638	&7.000	&2.193 \\
   \hline
     \multicolumn{1}{|l|}{50\%} & Simple & 5.933	&11.150	&5.900	&41.325	&1.030	&13.721	&2.004	&2.405 \\
     \multicolumn{1}{|l|}{}     & Tricky & 13.216	&379.238	&- &- &-	&32.363	&7.200	&2.331 \\ 
   \hline
     \multicolumn{1}{|l|}{75\%} & Simple & 3.654	&1.121	&3.729	&4.208	&1.015	&9.075	&1.483	&2.590 \\
     \multicolumn{1}{|l|}{}     & Tricky &8.672	&183.150	&- &- &-	&21.950	&5.800	&2.294 \\
   \hline
	 \hline
	  \multicolumn{2}{|c|}{Worst case}  &20	&600	&29 &600 &1.5	&56	&22	&4 \\
	 \hline
  \end{tabular}
  \vspace{.1cm}
 \end{center} \caption{Summary of results for instances with two nonbinary trees. The layout of the table is the same as that of table \ref{table:rspr_for_maf}.\label{table:maf_for_maf}}
\end{table}

\begin{table}
\begin{tabular}{|c|c|c|c|}
\hline
 & MIN & AVG & MAX \\ 
\hline
Common taxa & 3 & 5.235 & 20 \\ 
\hline
\emph{opt} & 0 & 0.873 & 7 \\ 
\hline
Ratio 4-approx & 1 & 1.002 & 1.2 \\ 
\hline
Ratio 6-approx & 1 & 1.088 & 3 \\ 
\hline
Gap (T-EST - MAF) & 0 & 0.010 & 1 \\ 
\hline
Gap (4-approx - MAF) & 0 & 0.020 & 2 \\ 
\hline
Time T-EST & 0 & 0.221 & 3 \\ 
\hline
Time 4-approx & 0 & 0.270 & 1 \\ 
\hline
\end{tabular}

\caption{Summary of results for dataset $F_1$ (204 gene trees) originally obtained from
GreenPhylDB database. Common taxa is the number of taxa after restricting the
gene tree and the species tree to common taxa. \emph{opt} is the exact hybridization number, as
computed by \textsc{TerminusEst}. Ratio 4-approx (resp. 6-approx) is the ratio of the solution obtained
by {\sc NonbinaryCycleKiller} (running in 4-approx, resp. 6-approx mode) to the solution obtained by \textsc{TerminusEst}. Gap (T-EST - MAF) is the absolute gap between the optimum
MAF solution (here computed with \textsc{rSPR}) and the exact hybridization number, as computed by \textsc{TerminusEst}. Gap (4-approx - MAF) is the absolute gap between
the optimum MAF solution and the reticulation number of the solution generated by
\textsc{NonbinaryCycleKiller} running in its 4-approx mode. Time T-EST is the running time (in seconds) of
\textsc{TerminusEst}, and Time 4-approx is the running time (in seconds) of \textsc{NonbinaryCycleKiller}
running in its 4-approx mode. In 202 instances \textsc{TerminusEst} returned
the same size solution as \textsc{rSPR}, in 202 cases \textsc{TerminusEst} returned the
same size solution as \textsc{NonbinaryCycleKiller} (running in 4-approx mode), and in 201 cases
\textsc{NonbinaryCycleKiller} (running in 4-approx mode) returned the same size solution as \textsc{rSPR}.}
\label{lab:f1}
\end{table}

\begin{table}[h]
\begin{tabular}{|c|c|c|c|}
\hline
 & MIN & AVG & MAX \\ 
\hline
Common taxa & 3 & 11.704 & 22 \\ 
\hline
\emph{opt} & 0 & 2.854 & 10 \\ 
\hline
Ratio 4-approx & 1 & 1.025 & 2 \\ 
\hline
Ratio 6-approx & 1 & 1.264 & 3 \\ 
\hline
Gap (T-EST - MAF) & 0 & 0.048 & 1 \\ 
\hline
Gap (4-approx - MAF) & 0 & 0.165 & 3 \\ 
\hline
Time T-EST & 0 & 0.576 & 7 \\ 
\hline
Time 4-approx & 0 & 0.605 & 3 \\ 
\hline
\end{tabular}
\caption{Summary of results for dataset $F_2$ (1003 gene trees) originally obtained from
GreenPhylDB database. In 955 instances \textsc{TerminusEst} returned
the same size solution as \textsc{rSPR}, in 911 cases \textsc{TerminusEst} returned the
same size solution as \textsc{NonbinaryCycleKiller} (running in 4-approx mode), and in 880 cases
\textsc{NonbinaryCycleKiller} (running in 4-approx mode) returned the same size solution as \textsc{rSPR}.}
\label{lab:f2}
\end{table}

\begin{table}
\begin{tabular}{|c|c|c|c|}
\hline
 & MIN & AVG & MAX \\ 
\hline
Common taxa & 2 & 14.206 & 22 \\ 
\hline
\emph{opt} & 0 & 3.613 & 12 \\ 
\hline
Ratio 4-approx & 1 & 1.027 & 2 \\ 
\hline
Ratio 6-approx & 1 & 1.277 & 3 \\ 
\hline
Gap (T-EST - MAF) & 0 & 0.065 & 2 \\ 
\hline
Gap (4-approx - MAF) & 0 & 0.195 & 4 \\ 
\hline
Time T-EST & 0 & 0.689 & 21 \\ 
\hline
Time 4-approx & 0 & 0.729 & 3 \\ 
\hline
\end{tabular}
\caption{Summary of results for dataset $F_3^p$ (5924 gene trees) originally obtained from
GreenPhylDB database. In 5553 instances \textsc{TerminusEst} returned
the same size solution as \textsc{rSPR}, in 5297 cases \textsc{TerminusEst} returned the
same size solution as \textsc{NonbinaryCycleKiller} (running in 4-approx mode), and in 5030 cases \textsc{NonbinaryCycleKiller} (running in 4-approx mode) returned the same size solution as \textsc{rSPR}.}
\label{lab:f3p}
\end{table}

\begin{table}
\begin{tabular}{|c|c|c|c|}
\hline
 & MIN & AVG & MAX \\ 
\hline
Common taxa & 3 & 17.319 & 22 \\ 
\hline
\emph{opt} & 0 & 1.560 & 12 \\ 
\hline
Ratio 4-approx & 1 & 1.021 & 2 \\ 
\hline
Ratio 7-approx & 1 & 1.704 & 4 \\ 
\hline
Gap (T-EST - MAF) & 0 & 0.053 & 4 \\ 
\hline
Gap (4-approx - MAF) & 0 & 0.132 & 5 \\ 
\hline
Time T-EST & 0 & 0.422 & 15 \\ 
\hline
Time 4-approx & 0 & 1.182 & 14 \\ 
\hline
\end{tabular}
\caption{Summary of results for dataset $F_3^s$ (5789 gene trees) originally obtained from
GreenPhylDB database. In 5552 instances \textsc{TerminusEst} returned
the same size solution as \textsc{rSPR}, in 5415 cases \textsc{TerminusEst} returned the
same size solution as \textsc{NonbinaryCycleKiller} (running in 4-approx mode), and in 5209 cases
\textsc{NonbinaryCycleKiller} (running in 4-approx mode) returned the same size solution as \textsc{MAF}. In
this dataset the gene trees were also nonbinary, meaning that \textsc{NonbinaryCycleKiller} had
to use the MAF algorithm described in \cite{nonbinCK} instead of \textsc{rSPR}.}
\label{lab:f3s}
\end{table}


%

\end{backmatter}
\end{document}